\documentclass[10pt]{article}

\usepackage{times}
\usepackage{fullpage}
\usepackage{amsfonts,amssymb,amsmath,amsthm}
\usepackage{latexsym}
\usepackage{url}
\usepackage[pdftex]{color}
\usepackage[breaklinks]{hyperref}

\newcommand{\OR}{\mathrm{OR}}
\newcommand{\norm}[1]{\|#1\|}
\newcommand{\PM}{\{\pm1\}}
\newcommand{\sq}{\mathcal{S}}
\newcommand{\eps}{\varepsilon}
\newcommand{\adeg}{\widetilde{\mathrm{deg}}}

\newtheorem{theorem}{Theorem}
\newtheorem{lemma}[theorem]{Lemma}
\newtheorem{corollary}[theorem]{Corollary}
\newtheorem{definition}[theorem]{Definition}
\newcommand{\thmref}[1]{\hyperref[#1]{{Theorem~\ref*{#1}}}}
\newcommand{\lemref}[1]{\hyperref[#1]{{Lemma~\ref*{#1}}}}
\newcommand{\corref}[1]{\hyperref[#1]{{Corollary~\ref*{#1}}}}
\newcommand{\defref}[1]{\hyperref[#1]{{Definition~\ref*{#1}}}}

\begin{document}
\title{A Dual Polynomial for OR}
\author{%
Robert {\v S}palek\\
Google, Inc.%
\thanks{Most of the work conducted while at CWI, Amsterdam, in February 2003.}\\
{\tt spalek@google.com}
}
\date{}
\maketitle

\begin{abstract}
We reprove that the approximate degree of the OR function on $n$ bits is $\Omega(\sqrt n)$.  We consider a linear program which is feasible if and only if there is an approximate polynomial for a given function, and apply the duality theory.  The duality theory says that the primal program has no solution if and only if its dual has a solution.  Therefore one can prove the nonexistence of an approximate polynomial by exhibiting a dual solution, coined the \emph{dual polynomial}.  We construct such a polynomial.
\end{abstract}

\section{Introduction}

We study the approximation of Boolean functions by real-valued polynomials.  This line of research was initiated by Minsky and Papert \cite{mp:perceptrons}.  An $n$-bit Boolean function $f$ is represented by a multivariate polynomial $p(x_1, \dots, x_n)$.
Nisan and Szegedy \cite{nisan&szegedy:degree} defined the approximate degree of a function $f$ under the $\ell_\infty$-norm, denoted $\adeg(f)$, as the smallest degree for which there exists a polynomial that is close to the function pointwise.  Several complexity measures have been since shown to be lower-bounded in terms of $\adeg(f)$: circuit size \cite{beigel:polmethod}, or quantum query complexity \cite{bbcmw:polynomialsj}.  Consider the OR function on $n$ bits.  Nisan and Szegedy \cite{nisan&szegedy:degree} showed that $\adeg(\OR_n) = \Theta(\sqrt n)$, and Paturi \cite{paturi:degree} extended their bound to all symmetric functions.

The existence of an approximate polynomial can be described by a linear program; let us coin it the primal program.  Using the duality theory of linear programming, one can show the non-existence of an approximate polynomial for a function $f$ by exhibiting a solution to its dual program, a so-called \emph{dual polynomial} for $f$.  Recently, several papers have appeared that use dual polynomials to prove good communication complexity lower bounds: Sherstov \cite{She07b} and Shi and Zhu \cite{sz:qcomm-block-composed-f} show two-party quantum communication lower bounds, and Lee and Shraibman \cite{ls:nof} and Chattopadhyay and Ada \cite{ca:disj-nof} show multi-party randomized communication lower bounds in the number-on-the-forehead model.  The basic idea of these papers is as follows.  One defines a special \emph{pattern matrix} (or tensor in the multi-party case) whose entries are values of a certain polynomial.  The structure of the pattern matrix allows one to relate properties of the polynomial to properties of the matrix, such as its trace norm.  The pattern matrix formed from the dual polynomial forms a witness to the large trace norm of the matrix.
The communication complexity is then lower-bounded in terms of the trace norm.  None of these papers actually presents an explicit dual polynomial for any function; they only use its existence and some inequalities guaranteed by the duality principle from the known bounds on the approximate degree.

It is natural to ask what a dual polynomial looks like for the simplest functions.  In this short note, we address this question and present an asymptotically optimal dual polynomial for the OR function.  Our proof extends the ideas of Buhrman and Szegedy \cite{bs:dualor}.

\section{Preliminaries}

\subsection{Symmetric polynomials}

We represent Boolean functions by polynomials in the Fourier basis, where $+1$ corresponds to the logical value $0$ (false) and $-1$ to the logical value $1$ (true).  In this basis, multiplication corresponds to the exclusive OR.
We say that $f:\PM^n \to \PM$ is a \emph{symmetric} function, if $f(x) = f(x_\sigma)$ for every permutation $\sigma \in S_n$ and $x \in \PM^n$, where $x_\sigma$ denotes a $\sigma$-permuted version of $x$, with $(x_\sigma)_i = x_{\sigma(i)}$.

Let $p: \PM^n \to \Re$ be a polynomial in variables $x_1, \dots, x_n$.  Since $x_i^2=1$, we can restrict ourselves to \emph{multilinear} polynomials, where each variable appears with degree at most 1.  We say that $p$ has \emph{degree} $d$ and \emph{pure high degree} $d'$, if each term in $p$ is a product of at most $d$ and at least $d'$ variables.
We say that $p$ is an \emph{$\eps$-approximation} for a function $f$, if $|p(x) - f(x)| \le \eps$ for every $x \in \PM^n$.
If $p$ is an $\eps$-approximation of a symmetric function $f$, then there exists a symmetric polynomial $p'$ with the same degree, pure high degree, and approximation factor: $p'(x) = \frac 1 {n!} \sum_{\sigma \in S_n} p(x_\sigma)$.

Let $[n] = \{0, 1, \dots, n\}$.  Given a symmetric function $f: \PM^n \to \PM$, one can define a single-variate function $F: [n] \to \PM$ such that $f(x) = F(|x|)$, where $|x| = \frac {n- (x_1 + \dots + x_n)} 2$ is the Hamming weight of $x$, i.e., the number of minuses in $x$.  Analogously, following \cite{mp:perceptrons}, given a symmetric multilinear polynomial $p: \PM^n \to \Re$, one can define a \emph{single-variate} polynomial $P: [n] \to \Re$ of the same degree such that
\begin{align*}
P(k) &= p(\underbrace{-1, \dots, -1}_k,\ \underbrace{+1, \dots, +1}_{n-k})
	&& \mbox{for all $k \in [n]$,} \\
p &= P\left(\frac {n - (x_1 + \dots + x_n)} 2\right)
	&& \mbox{mod $(x_1^2-1)$, \quad\dots,\quad mod $(x_n^2-1)$.}
\end{align*}
Note that the pure high degree of $p$ does not correspond to the smallest degree of a $k$-term in $P(k)$.  When we talk about the pure high degree of a single-variate polynomial, we mean the pure high degree of its corresponding multilinear polynomial.  

Let $p, q: \PM^n \to \Re$.  Define a \emph{scalar product} as $p \cdot q = \sum_{x \in \PM^n} p(x) q(x)$.  This induces a scalar product $P \cdot Q = \sum_{i=0}^n \binom n i P(i) Q(i)$ on the space of symmetric polynomials.  Similarly, the \emph{$\ell_1$-norm} $\norm p_1 = \sum_{x \in \PM^n} |p(x)|$ induces an $\ell_1$-norm $\norm P_1 = \sum_{i=0}^n \binom n i |P(i)|$.

Let $p: \PM^n \to \Re$ be a multilinear polynomial of degree $d$ and pure high degree $d'$, and consider $q(x) = p(x) \cdot (x_1 \cdots x_n) \mod (x_i^2-1)$.  In the functional interpretation, $q(x)$ equals \emph{$p(x)$ multiplied by the parity of $x$}.  Thanks to the term cancellation $x_i^2 = 1$, each term in $q$ corresponds to the complement of a term in $p$, and therefore $q$ has degree $n-d'$ and pure high degree $n-d$.
Now, assume that $p$ (and thus also $q$) are symmetric, and consider their corresponding single-variate polynomials $P, Q$.   Then $Q(k) = P(k) \cdot (-1)^k$, and
the degree of $P$ corresponds to $n$ minus the pure high degree of $Q$ and vice versa.

\subsection{Linear program for polynomial approximation}

\begin{theorem}
A total Boolean function $f: \PM^n \to \PM$ has $\eps$-approximate degree at least $d$ if and only if there exists a polynomial $b: \PM^n \to \Re$ with pure high degree $d$ such that $\frac {\norm b_1} {b \cdot f} < \frac 1 \eps$.
\end{theorem}

\begin{proof}
$f$ can be $\eps$-approximated by a polynomial of degree $d-1$ is equivalent to the feasibility of the following primal linear program.  Consider the \emph{Fourier basis} on the space of multilinear polynomials: $\{ \chi_S \}_{S \subseteq \{1, \dots, n\}}$, where $\chi_S(x) = \prod_{i \in S} x_i$.  Let $F = \{ \chi_S(x) \}_{x,S}$ denote the Fourier transform over $\mathbb{Z}_2^n$, indexed by $\PM^n$ and $S \subseteq \{1, \dots, n\}$, and let $a$ denote a vector of Fourier coefficients.
\begin{align*}
F a &\ge f - \eps \\
F a &\le f + \eps \\
a_S &= 0 \mbox{ for } |S| \ge d
\end{align*}

The primal program is unfeasible if and only if its dual is feasible.  The dual program is as follows.
\[
\begin{array}{r@{\ }l}
(b^+ - b^-) \cdot f &> (b^+ + b^-) \cdot \eps \\
(b^+ - b^-) F &= c \\
b^+, b^- &\ge 0 \\
c_S &= 0 \mbox{ for } |S| < d
\end{array}
\quad \Longleftrightarrow \quad
\begin{array}{r@{\ }l}
b \cdot f &> |b| \cdot \eps \\
b F &= c \\
c_S &= 0 \mbox{ for } |S| < d
\end{array}
\]
We can assume that $b^+$ and $b^-$ of the optimal solution are disjoint, i.e., $b^+(x) b^-(x) = 0$ for each $x$, otherwise we could lower the right-hand side of the first inequality by subtracting the same constant $\min(b^+(x), b^-(x)) > 0$ from both $b^+(x)$ and $b^-(x)$, and the remaining expressions would stay unchanged.  Let $b = b^+ - b^-$ and $|b| = b^+ + b^-$.  The constraints $b F = c$ and $c_S = 0$ for $|S| < d$ say that $b$ has pure high degree $d$.  The dual is feasible if and only if there exists such a $b$ with $b \cdot f > |b| \cdot \eps = \eps \norm b_1$.
\end{proof}

Note that if $f$ is symmetric, then it suffices to look for a dual polynomial $b$ in the space of symmetric polynomials.  Let us reformulate the condition in the language of single-variate polynomials.

\begin{corollary} \label{cor:dualsym}
A total symmetric Boolean function $F: [n] \to \PM$ has $\eps$-approximate degree at least $d$ if and only if there exists a polynomial $B: [n] \to \Re$ with pure high degree $d$ such that $\frac {\norm B_1} {B \cdot F} < \frac 1 \eps$.
\end{corollary}

\section{Dual polynomial for OR}

%Let $\OR(+1, \dots, +1)=+1$ and $\OR(x)=-1$ otherwise.
First, we define a certain low-degree polynomial $P$ and show that its norm $\norm P_1$ is not too large compared to its value $P(0)$.  This polynomial will be crucial for defining the dual polynomial for OR.  The design of our polynomial comes from extending the ideas of Buhrman and Szegedy \cite{bs:dualor}.

\begin{definition} \label{def:P}
Let $m = \lfloor \sqrt n \rfloor$ and
let $\sq = \{ i^2 : i \in [m] \} \cup \{ 2 \}$ denote the set containing the integer squares up to $n$ and the number $2$.
Define a polynomial
\[
P(x) =
	2 (-1)^{n-m-1} \frac {m!^2} {n!} \cdot
	\prod_{i \in [n] - \sq} (x-i) \enspace.
\]
\end{definition}

The multiplicative factor of $P$ is chosen such that $P(0) = 1$.  The degree of $P$ is $n-m-1$.

\begin{lemma}
\label{lem:binom-bound}
For every pair of integers $k, m$ with $k \le m$, $\frac {m!^2} {(m+k)! (m-k)!} \le
1$.
\end{lemma}

\begin{proof}
The term is a product of numbers that are all smaller than 1:
\[
\frac {m!^2} {(m+k)! (m-k)!} =
	\frac {m (m-1) \dots (m-k+1)} {(m+k) (m+k-1) \dots (m+1)} =
	\prod_{i=1}^k \left(1 - \frac k {m+i} \right) \le 1
\qedhere
\]
\end{proof}

\begin{lemma}
\label{lem:pbound}
$\binom n 2 |P(2)| \le 12$ and $\binom{n}{k^2} |P(k^2)| \le \frac 8 {k^2}$ for every $k=1,2,\dots,m$.
\end{lemma}

\begin{proof}
First, we substitute $x=2$ into $|P(x)|$ and rewrite the product over $i \in [n] - \sq$ as the ratio of two products, one over $i \in [n] - \{0,1,2\}$ and one over $i \in \sq - \{0,1,2\}$.  We then pull the $j=2$ term out of the product in the denominator, use $|j^2-2| < j^2-4$, and apply \lemref{lem:binom-bound}.
\begin{align*}
|P(2)| &= 2 \frac {m!^2} {n!} \frac {(n-2)!} {2 \prod_{j=3}^m |2-j^2|} <
	\frac {m!^2} {n!} \frac {(n-2)!} {\prod_{j=3}^m (j^2-4)} \\
&= \frac {m!^2} {n!} \frac {(n-2)!} {\prod_{j=3}^m (j+2) (j-2)} =
	\frac 1 {n(n-1)} \frac {m!^2} {\frac {(m+2)!} {4!} (m-2)!}
\le \frac {4!} {n(n-1)} = \frac {12} {\binom n2}.
\end{align*}

Second, we substitute $x=k^2$ to $|P(x)|$ and rewrite the product over $i \in [n] - \sq$ as the ratio of two products, one over $i \in [n] - \{k^2\}$ and one over $i \in \sq - \{k^2\}$. The term $i=k^2$ does not appear in any of products, because it is $0$.
\begin{align*}
|P(k^2)| &= 2\frac {m!^2} {n!} \cdot
	\frac {\prod_{\substack{i \in[n] \\ i \ne k^2}} |k^2-i|}
	{|k^2-2| \cdot \prod_{\substack{j \in [m] \\ j \ne k}} (k+j) |k-j|} \\
&= 2 \frac {m!^2} {n!} \cdot
	\frac {k^2! (n-k^2)!}
	{\frac {\displaystyle (k+m)!} {\displaystyle  2k \cdot (k-1)!}
	\cdot k! (m-k)!} \cdot \frac 1 {|k^2-2|} \\
&= 4 \cdot \frac {k^2! (n-k^2)!} {n!} \cdot \frac {m!^2} {(m+k)! (m-k)!}
	\cdot \frac 1 {|k^2-2|} \\
\intertext{Apply \lemref{lem:binom-bound} and $|k^2-2| \ge k^2/2$, which holds for all integers $k \ge 1$.}
&\le \frac 4 {\binom n {k^2}} \cdot \frac 1 {|k^2-2|}
	\le \frac 4 {\binom n {k^2}} \cdot \frac 1 {k^2/2}
	\le \frac 8 {\binom n {k^2} k^2}.
\end{align*}
Note that if we did not include the number 2 into $\sq$, in \defref{def:P}, then the upper bound on $|P(k^2)|$ would be much weaker, without the factor of $1/k^2$.
\end{proof}

Now we show that a constant fraction of the norm of $P$ comes from the term $P(0)=1$.

\begin{theorem}
\label{thm:sum}
$\norm P_1 < 27$.
\end{theorem}

\begin{proof}
First, use the fact that $P(i)=0$ for $i \in [n]-\sq$, non-square integers $i$ other than 2.
\begin{align*}
\norm P_1
&= \sum_{i=0}^n \binom n i |P(i)| = \sum_{i \in \sq} \binom n i |P(i)| \\
&= P(0) + \binom n 2 P(2) + \sum_{k=1}^m \binom n {k^2} |P(k^2)| \\
\intertext{Now, use $P(0)=1$, \lemref{lem:pbound}, and $\sum_k \frac 1 {k^2} = \frac {\pi^2} 6$.}
&\le 13 + 8 \sum_{k=1}^m \frac 1 {k^2}
< 13 + 8 \frac {\pi^2} 6 < 27.
\qedhere
\end{align*}
\end{proof}

Finally, we are ready to present the dual polynomial for OR.

\begin{theorem}
The $\frac 1 {14}$-approximate degree of OR on $n$ bits is at least $\sqrt n$.
\end{theorem}

\begin{proof}
Consider the polynomial
\[
Q(k) = (-1)^k P(k) \enspace,
\]
that is $P$ from \defref{def:P} multiplied by parity.  We show that $Q$ is a good dual polynomial for OR.  First, the pure high degree of $Q$ is $n-(n-m-1) = m+1 > \sqrt n$.  Second, we compute the ratio from \corref{cor:dualsym}.
Since $\OR(0)=1$ and $\OR(k)=-1$ for $k \ge 1$, $Q \cdot \OR = 2 Q(0) - Q \cdot 1 = 2 Q(0)$, because $Q$ has no constant coefficient.  Now, we use \thmref{thm:sum} to upper-bound the numerator and conclude
\[
\frac {\norm{Q}_1} {Q \cdot \OR}
	= \frac {\norm{P}_1} {2 P(0)}
	< \frac {27} 2
	< 14 \enspace.
\qedhere
\]
\end{proof}

\section{Open problems}

%It is interesting to compare the lower bound obtained by the polynomial method to the lower bound obtained by the quantum adversary method.  Both give an asymptotically tight bound $\Omega(\sqrt n)$ for bounded-error quantum algorithms computing OR, but each for a completely different reason.  The adversary bound is information theoretical based on the distinguishability of the inputs, and one gets a tight bound even we are promised that the input contains at most one 1.  The polynomial lower bound on the other hand relies on the fact that the acceptance probability must be between $0$ and $1$ for each Hamming weight of the input string.  If we were given the same promise as above, then OR can be computed exactly by a degree-1 polynomial $x_1 + \dots + x_n$.

The approximate degree of the $t$-threshold function on $n$ bits is $\Theta(\sqrt{t (n-t)})$ \cite{paturi:degree}.  It would be interesting to find an explicit dual polynomial for the threshold function.  A good candidate may be $Q(k) = (-1)^k P(k)$ with
\[
P(x) = \prod_{i \in [n]-T} (x-t-i) \enspace,
\]
where $T$ is a set of integers that can be written as $k^2 - \ell^2$, where $k \in [\lfloor \sqrt{n-t}\rfloor]$ and $\ell \in [\lfloor \sqrt t \rfloor]$.  Note that $|T| = \Theta(\sqrt{t (n-t)})$.

The approximate degree of the two-level AND-OR tree on $n$ bits (with all gates of fan-in $\sqrt n$) is only known to lie between $O(\sqrt n)$ and $\Omega(\sqrt[3] n)$.  Both bounds have been obtained through quantum algorithms, as follows.
Consider a $T$-query quantum algorithm.  Its acceptance probability on input $x$ can be expressed as a $2 T$-degree polynomial $p$ in the variables $x_1, \dots, x_n$ \cite{bbcmw:polynomialsj}.  If the algorithm computes a function $f$ with bounded error, then $p$ approximates $f$.  Therefore quantum algorithms give approximate polynomials, and approximate degree lower bounds give quantum query lower bounds.
For the two-level AND-OR tree, the upper bound is via a quantum search algorithm on noisy inputs \cite{hmw:berror-search} and the lower bound is via a reduction from the element distinctness problem \cite{as:collision}.  Can one compute the approximate degree of the AND-OR tree by showing a good dual polynomial?

\section*{Acknowledgments}

We thank Harry Buhrman and Mario Szegedy for starting the project, coming up with the crucial ideas, and many fruitful discussions.  We also thank Ronald de Wolf for fruitful discussions, and Troy Lee for proofreading.

\bibliographystyle{alpha}
\bibliography{../quantum}

\end{document}